\documentclass[12pt]{amsart}

\usepackage[utf8]{inputenc}
\usepackage[T1]{fontenc}

\usepackage{a4wide}
\usepackage{amsmath,amssymb,amsthm,courier}
\usepackage{enumerate}
\usepackage{parskip}
\usepackage[pdftex]{graphicx}
\usepackage{verbatim}
\usepackage{hyperref}
\usepackage{color}
\usepackage{bbm}

\usepackage{soulutf8} 

\usepackage[cm]{fullpage}

\newtheorem{theorem}{Theorem}[section]

\makeatletter
\@addtoreset{equation}{section}

\newcommand{\bes}{\begin{displaymath}}
\newcommand{\ees}{\end{displaymath}}
\newcommand{\be}{\begin{equation}}
\newcommand{\ee}{\end{equation}}
\newcommand{\ba}{\begin{eqnarray}}
\newcommand{\ea}{\end{eqnarray}}
\newcommand{\bas}{\begin{eqnarray*}}
\newcommand{\eas}{\end{eqnarray*}}
\newcommand{\@Bbb}[1]{\ensuremath{\Bbb #1}}

\newcommand{\B}{{\@Bbb B}}
\newcommand{\E}{{\@Bbb E}}
\newcommand{\F}{{\@Bbb F}}
\renewcommand{\P}{{\@Bbb P}}

\newcommand{\Q}{{\@Bbb Q}}
\newcommand{\bQ}{{\@Bbb Q}}
\newcommand{\N}{{\@Bbb N}}
\newcommand{\R}{{\@Bbb R}}
\newcommand{\T}{{\@Bbb T}}
\newcommand{\bbR}{{\@Bbb R}}
\newcommand{\W}{{\@Bbb W}}
\newcommand{\Z}{{\@Bbb Z}}
\newcommand{\bbZ}{{\@Bbb Z}}

\newcommand{\@s}[1]{\ensuremath{\mathcal #1}}
\newcommand{\cA}{\@s A}
\newcommand{\cB}{\@s B}
\newcommand{\cC}{\@s C}
\newcommand{\cD}{\@s D}
\newcommand{\cE}{\@s E}
\newcommand{\cF}{\@s F}
\newcommand{\cG}{\@s G}
\newcommand{\cH}{\@s H}
\newcommand{\cI}{\@s I}
\newcommand{\cJ}{\@s J}
\newcommand{\cal}{\mathcal}
\newcommand{\cK}{\@s K}
\newcommand{\cL}{\@s L}
\newcommand{\cN}{\@s N}
\newcommand{\cM}{\@s M}
\newcommand{\cO}{\@s O}
\newcommand{\cP}{\@s P}
\newcommand{\cQ}{\@s Q}
\newcommand{\cR}{\@s R}
\newcommand{\cS}{\@s S}
\newcommand{\cT}{\@s T}
\newcommand{\cU}{\@s U}
\newcommand{\cV}{\@s V}
\newcommand{\cW}{\@s W}
\newcommand{\cX}{\@s X}
\newcommand{\cY}{\@s Y}
\newcommand{\cZ}{\@s Z}

\def\d{{\rm d}}

\def\i{{\rm i}}

\newcommand{\@bm}[1]{\ensuremath{\mathbf #1}}
\newcommand{\bma}{\@bm a}\newcommand{\bmA}{\@bm A}
\newcommand{\bmb}{\@bm b}\newcommand{\bmB}{\@bm B}
\newcommand{\bmc}{\@bm c}\newcommand{\bmC}{\@bm C}
\newcommand{\bmd}{\@bm d}\newcommand{\bmD}{\@bm D}
\newcommand{\bme}{\@bm e}
\newcommand{\bmf}{\@bm f}\newcommand{\bmF}{\@bm F}
\newcommand{\bmg}{\@bm g}\newcommand{\bmG}{\@bm G}
\newcommand{\bmh}{\@bm h}\newcommand{\bmH}{\@bm H}
\newcommand{\bmi}{\@bm i}\newcommand{\bmI}{\@bm I}
\newcommand{\bmj}{\@bm j}
\newcommand{\bmk}{\@bm k}\newcommand{\bmK}{\@bm K}
\newcommand{\bml}{\@bm l}
\newcommand{\bmm}{\@bm m}\newcommand{\bmM}{\@bm M}
\newcommand{\bmn}{\@bm n}
\newcommand{\bmo}{\@bm o}
\newcommand{\bmp}{\@bm p}
\newcommand{\bmq}{\@bm q}\newcommand{\bmQ}{\@bm Q}
\newcommand{\bmr}{\@bm r}
\newcommand{\bms}{\@bm s}\newcommand{\bmS}{\@bm S}
\newcommand{\bmt}{\@bm t}
\newcommand{\bmu}{\@bm u}\newcommand{\bmU}{\@bm U}
\newcommand{\bmw}{\@bm w}\newcommand{\bmW}{\@bm W}
\newcommand{\bmv}{\@bm v}\newcommand{\bmV}{\@bm V}
\newcommand{\bmx}{\@bm x}\newcommand{\bmX}{\@bm X}\newcommand{\bx}{\@bm x}
\newcommand{\bmy}{\@bm y}\newcommand{\bmY}{\@bm Y}\newcommand{\by}{\@bm y}
\newcommand{\bmz}{\@bm z}\newcommand{\bmZ}{\@bm Z}

\newcommand{\bmzero}{\@bm 0}

\newcommand{\@g}[1]{\ensuremath{\mathfrak #1}}
\newcommand{\gA}{\@g A}
\newcommand{\gD}{\@g D}
\newcommand{\gJ}{\@g J}
\newcommand{\gF}{\@g F}
\newcommand{\gM}{\@g M}
\newcommand{\gR}{\@g R}

\title[Convex set of PPT states]
{\large Convex set of quantum states \\
  with positive partial transpose \\ 
   analysed by hit and run algorithm}

\author{Konrad Szyma{\'n}ski}
\address{K.S: Institute of Physics, Jagiellonian University, Cracow, Poland}
\email{konrad.szymanski@uj.edu.pl}

\author{Beno\^\i{}t Collins}
\address{B.C.: Department of Mathematics, Graduate School of Science,
Kyoto University, Kyoto 606-8502, Japan
and
CNRS, 
France}
\email{collins@math.kyoto-u.ac.jp}

\author{Tomasz Szarek}
\address{T.S.: 
Faculty of Physics and Applied Mathematics, Gda{\'n}sk University of Technology,
ul. Gabriela Narutowicza 11/12, 80-233 Gda{\'n}sk, Poland
}
\email{szarek@intertele.pl}

\author{Karol {\.Z}yczkowski}
\address{K.{\.Z}: Institute of Physics, Jagiellonian University, Cracow, Poland
                                and
  Center for Theoretical Physics, Polish Academy of Sciences, Warsaw}
\email{ karol@tatry.if.uj.edu.pl}

\begin{document}


\begin{abstract}
The convex set of quantum states of a composite $K \times K$ system
with positive partial transpose is analysed. 
A version of the {\sl hit and run} algorithm 
is used to generate a sequence of random points covering this set uniformly
and an estimation for the convergence speed of the algorithm is derived.
For $K\ge 3$ this algorithm works faster than sampling over the entire
set of states and verifying whether the partial transpose is positive.
The level density of the PPT states is shown to differ from
the Marchenko-Pastur distribution, supported in $[0,4]$
and corresponding asymptotically to the entire set of quantum states. 
Based on the shifted semi--circle law, describing 
asymptotic level density of partially transposed states,
and on the level density for the Gaussian unitary ensemble
with constraints for the spectrum we find an explicit form
of the probability distribution supported in $[0,3]$, 
which describes well 
the level density obtained numerically for PPT states.
\end{abstract}

\date{March 28, 2017}

\maketitle

\section{Introduction}

The structure of the set of quantum states is a subject of a current interest
form the perspective of pure mathematics \cite{AS03} and theoretical physics \cite{ACH93,KZ01,GMK05,BZ06}.
A special attention is paid to the bipartite systems, in which separable and entangled states
can be distinguished \cite{KZ01} as entanglement plays a key role in the theory of quantum information processing \cite{H^4}.

It is well known that any separable state of a bipartite, $K \times K$ system
has positive partial transpose (PPT).
In the case of a two--qubit system, $K=2$, every PPT state is separable, while 
for higher systems it is not the case \cite{HHH96a}.
Description of the set of separable states for an arbitrary $K$ is difficult,
while it is straightforward for the set of PPT states.

Up till now, we are not aware of any dedicated algorithm to obtain a random PPT states
with respect to the flat (Hilbert--Schmidt) measure. 
It is not difficult to generate a random density matrix 
from the entire set $\Omega_N$ of quantum states of size $N$
using a matrix $G$ of the Ginibre ensemble of complex square 
random matrices of order $N$,
as the matrix  $\rho=GG^{\dagger} / {\rm Tr} GG^{\dagger}$
is positive and normalized and forms thus 
a legitimate random quantum state \cite{ZS01,ZPNC11}.
However, as the relative volume of the  subset of PPT states
decreases exponentially with the dimension \cite{ZHSL98,GB02,AS06},
is not efficient to generate random points in the entire set 
$\Omega_N$ of states of size $N=K^2$
and to check a posteriori, if the PPT property is satisfied.

The goal of this work is to analyse the set of PPT states by generating
a sequence of random points from this set according to the flat measure.
In this way we are in position to analyse the level density of 
random PPT states and to show the difference with respect of the 
Hilbert--Schmidt density, which corresponds to uniform sampling
of the entire set of quantum states and asymptotically tends to the 
Marchenko-Pastur distribution.
Note that analytical analysis of the set of PPT states is not an easy task,
as this set does not enjoy the symmetry with respect to unitary transformations.
However, based
on the fact that the level density of partially transposed
quantum states is asymptotically described by the shifted semi--circle law
\cite{Au12,BTL12} and making use of the level density of for the
Gaussian Unitary Ensemble (GUE) with constraints for the spectrum \cite{DM08},
we are in position to conjecture the level density which 
describes well  numerical results obtained by sampling 
the set of PPT states.

A parallel aim of the paper is to introduce an effective 
algorithm of generating random points in 
a convex high dimensional set and to analyse its features. We propose a variant of the 
{\sl hit and run} algorithm \cite{Lo93,Ve05,LV1,LV2} and establish its speed of convergence. 
In each iteration of the algorithm one selects a random direction by choosing a random point from a hypersphere. 
In the next step one finds the longest interval belonging to the set which
passes through the initial point and is parallel to this direction and 
jumps to a random point drawn from this interval.

Results obtained are easily applicable for all convex sets for which it is computationally 
efficient to find two points in which a given line enters and exists the set.
Hence this approach, suitable for the set of PPT states,
is not easy to apply for the set of separable states of higher dimensions, 
for which  the problem
of finding out whether a given state is separable is  `hard' \cite{Gu04}.

This paper is organised as follows.
In section \ref{section:two} we  present the algorithm and provide 
a general estimation of the convergence speed.
Usage of the algorithm for simple bodies in  ${\mathbb R}^d$
is illustrated in section 3.
In section 4 we discuss the set of quantum states of composite systems  
and its subset containing the states with positive partial transform.
Making use of the algorithm we generate sequences of random points
in this set and analyze numerical results for the level 
density for such states. Considerable deviations from the known
distributions for the entire set of states are reported
and it is conjectured that the asymptotic level density
for the PPT states is described by 
the distribution the (\ref{eq:sp}), 
different from the Marchenko--Pastur distribution \cite{MP67}.

\section{Hit and run algorithm and its convergence speed}
\label{section:two}

Consider a given convex, compact set   $X \subset {\mathbb R}^d$. 
To generate a sequence of random points distributed according 
to the uniform Lebesgue measure in $X$ we will use the following 
version of the 
``hit and run" algorithm \cite{Ve05,CKKSZ13}.

\begin{enumerate}
\item
 Choose an arbitrary starting point $x_0 \in X$,
\item
    Draw randomly a unit vector $e \in S^{d-1}$  according to the uniform  distribution on the hypersphere,
\item
Find boundary points  $x_1^{\rm min}(e), x_1^{\rm max}(e) \in \partial X$
along the direction $e$:
there exist $a,b\ge 0$ such that
$x_1^{\rm min}(e)=x_0-a e$ and $ x_1^{\rm max}(e)+be$.
\item
Select a point $x_1$  randomly with respect to the
    uniform measure in the interval $[x_1^{\rm min}, x_1^{\rm max}]$.
\item
Repeat the steps (ii)-(iv) to find subsequent random points
  $x_2,x_3,\dots$.
\end{enumerate}

Note that each time the direction $e$
is taken randomly and that the direction used in step $i+1$
does not depend on the direction used in step $i$.

A first result of this work consists in the following
estimation of the convergence speed in the total variation norm.
The total variation norm of two probability measures is  defined 
by the formula
$$
\|\mu_1-\mu_2\|_{TV}:={\hat\mu}^+(X)+\hat{\mu}^-(X),
$$
where $\hat{\mu}^+-\hat{\mu}^-$ is the Jordan decomposition of the signed measure $\mu_1-\mu_2$.

Let $\mathbf{\Phi}$ be a Markov chain \cite{Has}
described by the above algorithm and let $S$ be its transition kernel. 

\begin{theorem}
\label{convergence-2}
 Let $X\subset\mathbb R^d$ be a compact convex set. Assume that there exist $r,R >0$
such that for some $x_0 \in X$ one has
$B(x_0, r)\subset X\subset B(x_0, R)$.
Then the chain $\mathbf\Phi$
 satisfies the condition
 \begin{equation}\label{e211}
 \|S^n(x, \cdot)-\mathcal L_d\|_{TV}\le(1-\theta)^n\quad\text{for any $x\in X$},
 \end{equation} 
 where
\begin{equation}
\label{e21}
\theta=(2/d) [(R/r+1)^{d-1}(R/r)]^{-1}. 
\end{equation}
(Here $\mathcal L_d$ denotes the $d$--dimensional Lebesgue measure.)

\end{theorem}

\begin{proof} At the very beginning of the proof we show that the Lebesgue measure $\mathcal L_d$ restricted to $X$ is invariant for the considered algorithm. To do this observe that for any $x\in X\setminus\partial X $ and a Borel set $A\subset X$ we have
$$
S(x, A)=\Upsilon\int_{S^{d-1}}\frac{1}{\|x^{\rm max}(e)-x^{\rm min}(e)\|}\int_{\mathbb R}{\bf
1}_A(x+\lambda e)\d \lambda \mathcal L_{d-1}(\d e),
$$
where 
$$
\Upsilon= (\mathcal L_{d-1}(S^{d-1}))^{-1}.
$$
Therefore, by the Fubini theorem we have
$$
\begin{aligned}
\int_X S(x, A)\mathcal L_d(\d x)&=\int_X \Upsilon\int_{S^{d-1}}\frac{1}{\|x^{\rm max}(e)-x^{\rm min}(e)\|}\int_{\mathbb R}{\bf
1}_A(x+\lambda e)\d \lambda \mathcal L_{d-1}(\d e)\mathcal L_d(\d x)\\
&=\Upsilon\int_{S^{d-1}}\frac{1}{\|x^{\rm max}(e)-x^{\rm min}(e)\|}\int_{\mathbb R}\int_X{\bf
1}_A(x+\lambda e)\mathcal L_d(\d x)\d \lambda \mathcal L_{d-1}(\d e)\\
&=\Upsilon\int_{S^{d-1}}\frac{1}{\|x^{\rm max}(e)-x^{\rm min}(e)\|}\int_{\mathbb R}\int_X{\bf
1}_{A-\lambda e}(x)\mathcal L_d(\d x)\d \lambda \mathcal L_{d-1}(\d e)\\
&=\Upsilon\int_{S^{d-1}}\frac{1}{\|x^{\rm max}(e)-x^{\rm min}(e)\|}\int_{\mathbb R}\mathcal L_d (A-\lambda e)\d \lambda \mathcal L_{d-1}(\d e)\\
&=\mathcal L_d (A)\Upsilon\int_{S^{d-1}}\frac{1}{\|x^{\rm max}(e)-x^{\rm min}(e)\|}\int_{\mathbb R}\d \lambda \mathcal L_{d-1}(\d e)=\mathcal L_d (A)\Upsilon \Upsilon^{-1}=\mathcal L_d (A),
\end{aligned}
$$
which finishes the proof that the Lebesgue measure $\mathcal L_d$ is invariant.

We are going to evaluate $S (x, \cdot)$ for any $x\in X$. First assume
 that $X=B(x, R)$. It is easy to see that the transition function $S(x, \cdot)$
 is absolutely continuous with respect to the Lebesgue measure $\mathcal L_d$. 
Its density is equal to
\begin{equation}\label{e1.5.6.13}
f_x(u)=(c_d \|u-x\|^{d-1} R)^{-1}\qquad\text{for $u\in B(x, R)\setminus \{x\}$,}
\end{equation}
where $c_d=2\pi^{d/2}/\Gamma(d/2)$ denotes the surface of the $d$-sphere of radius $1$.
Indeed, by symmetry argument we see that $f_x(u)=\tilde f (\|u-x\|)$ for some function $\tilde f: [0, R]\to \mathbb R$. Then we have
$$
\int_0^u \tilde f (v) c_d v^{d-1} \d v=u/R\qquad\text{for any $u\in [0, R]$.}
$$
Differentiating both sides with respect to $u$ we obtain (\ref{e1.5.6.13}).

Now let $X$ be arbitrary. Since $X\subset B(x_0, R)$ we see that for any $x\in X$ the transition function $S(x, \cdot)$ is absolutely continuous with respect to the Lebesgue measure $\mathcal L_d$ and its density $f_x$ satisfies
$$
f_x(u)\ge (c_d \|u-x\|^{d-1} 2R)^{-1}\quad\text{for $u\in X$.}
$$
Thus
$$
\begin{aligned}
S(x, \cdot)&\ge \int_{\cdot\cap B(x_0, r)} f_x(u)\mathcal L_d(d u)
\ge (c_d (R+r)^{d-1} 2R)^{-1} \mathcal L_d(B(x, r))\nu (\cdot)\\
&\ge (b_d/c_d) [(R/r+1)^{d-1}(R/r)]^{-1}\nu (\cdot),
\end{aligned}
$$
where $\nu(\cdot)=\mathcal L_d (\cdot\cap B(x_0, r))/\mathcal L_d (B(x_0, r))$
and $b_d=\pi^{d/2}/\Gamma(d/2+1)$
denotes the volume of a unit ball in $\mathbb R^d$. Since $b_d/c_d=2/d$, we finally obtain
$$
S(x, \cdot)\ge (2/d) [(R/r+1)^{d-1}(R/r)]^{-1}\nu(\cdot).
$$
Now the version of Doeblin's theorem finishes the proof -- see Proposition 3.1 in \cite{CKKSZ13} (see also \cite{Doeblin}).
\end{proof}

\section{Balls, cubes and simplices in $\mathbb R^d$.}

To show the convergence estimate (\ref{e211}) in action,
in this section we apply it for simple bodies in ${\mathbb R}^d$.

For cubes and balls in ${\mathbb R}^d$ it is
straightforward to generate random points according to the uniform measure,
so we will not advocate to use the above algorithm for this purpose.
However, it is illuminating to compare estimations
for the parameters determining the convergence rate
according to Eq. (\ref{e211}).

\medskip

 For a unit {\bf ball} $B^d$ both radii coincide, $R=r$,
so their ratio $\kappa =r/R$ is equal to unity.
Estimation (\ref{e21}) gives $\theta=\; 2^{-d}/d$
where $b_d$ denotes the volume of a unit $d$--ball.
This implies the convergence rate 
 $$\alpha=1-\theta=1- 2^{-d}/d.$$

For an unit {\bf cube} $C^d$ the inscribed radius $r=1/2$,
and outscribed radius $R=\frac{1}{2}\sqrt{d}$
so the ratio  reads $\kappa=r/R=1/\sqrt{d}$. 
 Hence the convergence rate reads,
 $$\alpha=1-\theta=1-(2/d)[(\sqrt{d}+1)^{d-1}\sqrt{d}]^{-1}.$$

To demonstrate usefulness of the algorithm 
we present in Fig. \ref{circ-squa}  the density of points 
it generated in a circle and a square.

The $\chi^2$ test of the data obtained numerically 
does not suggests to reject the  hypothesis that the data were generated according to the uniform distribution with a confidence level $p=0.999$. 
\begin{figure}
\centering
(a)\raisebox{-6.5cm}{\includegraphics[height=7cm]{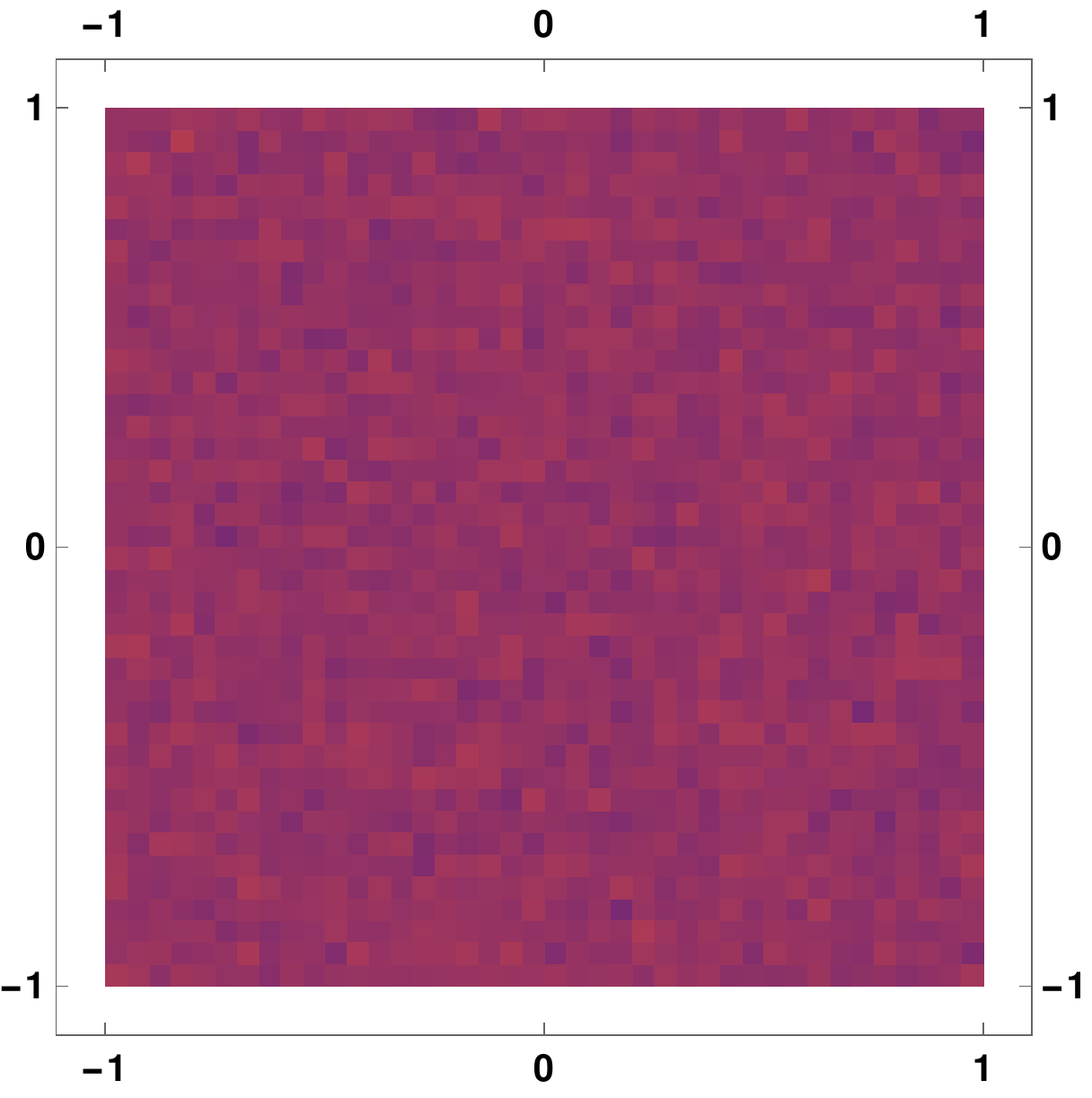}\includegraphics[height=7cm]{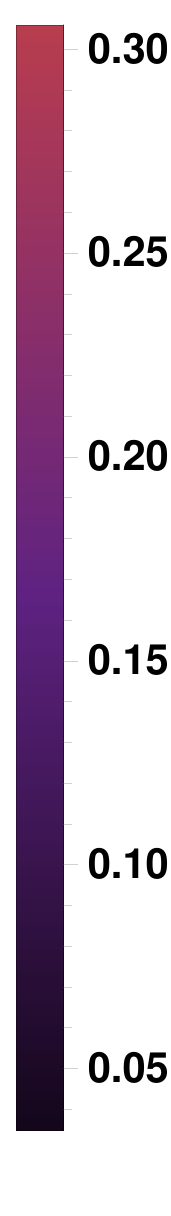}
}
(b)\raisebox{-6.5cm}{\includegraphics[height=7cm]{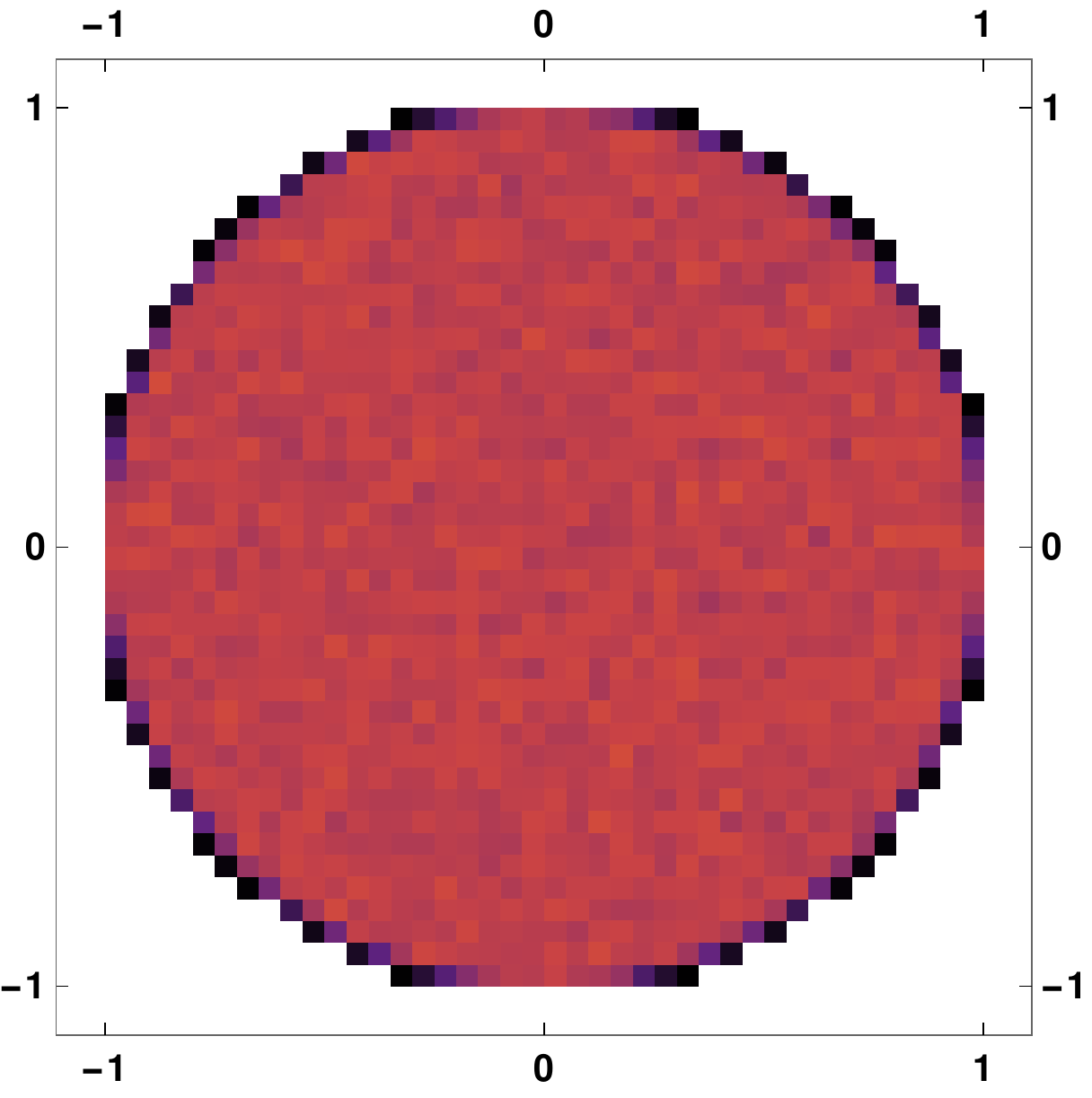}\includegraphics[height=7cm]{circle-desc.pdf}}
\caption{Histograms of sampled points in two simple 2--D cases: (a) square and (b) unit circle. For both plots the number of sampled points is $500 000$, the bins are squares of side $0.05$.
}
\label{circ-squa}
\end{figure}


For an $N$--{\bf simplex} $\Delta_{N}$ embedded in
${\mathbb R}^d$ with $d=N-1$ we have $R=\sqrt{(N-1)/N}$
and $r=1/\sqrt{N(N-1)}$ so that the ratio $\kappa=1/(N-1)=d^{-1}$.
In this case the convergence rate reads,
$$\alpha=1-\theta= 1-(2/d)[(d+1)^{d-1} d]^{-1}.$$

Note that the simplex $\Delta_{N}$ describes the set of classical states --
$N$-point probability distributions.

\section{Quantum states with positive partial transpose}

Let $\Omega_N$ be the set of density matrices of size $N$, 
formed by complex Hermitian and positive operators, $\rho^{\dagger}=\rho \ge 0$,
normalized by the trace condition Tr$\rho=1$.
Due to the normalization 
condition the dimension of the set is  $d=N^2-1$.

The radius of the out-sphere, equal to the Hilbert--Schmidt distance
between a pure state diag$(1,0,\dots,0)$
and the maximally mixed state $\rho_*={\mathbb I}/N$,
is $R=\sqrt{(N-1)/N}$.
The radius of the inscribed sphere given by the distance
between $\rho_*$ and the center of a face,
diag$(0,1,\dots,1)/(N-1)$ is equal to
$r=1/\sqrt{N(N-1)}$, so their ratio
 $\kappa=1/(N-1) =1/(\sqrt{d+1}-1) \sim d^{-1/2}$.
Hence the convergence rate of the algorithm reads in this case 
$$\alpha=1-\theta\sim 1-(2/d)[(\sqrt{d}+1)^{d-1}\sqrt{d}]^{-1}.$$

Consider now a composite dimension $N=K^2$, 
so the Hilbert space has a tensor product structure, 
${\cal H}_N={\cal H}_A\otimes {\cal H}_B$.
In this case one defines a partial transposition, 
${\mathbf T}_2={\mathbb I} \otimes {\mathbf T}$
and the set of PPT states (positive partial transpose)
which satisfy $\rho^{{\mathbf T}_2} \ge 0$. 
This set forms a convex subset of
$\Omega_N$ and it can be obtained as an intersection of
$\Omega_N$ and its reflection ${\mathbf T}_2(\Omega_N)$ \cite{BZ06}.
The set of PPT states can be decomposed into cones of the same height
$r=1/\sqrt{N(N-1)}$ \cite{SBZ06},
 hence $\mu=1/(N-1)\sim d^{-1/2}$.

In the case of set of PPT states the algorithm to generate
random elements from this set according to the flat measure
gives the same rate of convergence as in the case of quantum states,
$\alpha=1-\theta\sim 1-(2/d)[(\sqrt{d}+1)^{d-1}\sqrt{d}]^{-1}.$

\subsection{Two--qubit system}

   We start the discussion with the simplest composed
system consisting of two qubits, setting $K=2$
and $N=K^2=4$. To construct a random state 
from the entire $15$ dimensional set $\Omega_4$ 
we generated a random square Ginibre matrix $G$
of order four with independent complex Gaussian entries
and used the normalized Wishart matrix,
writing $\rho=GG^{\dagger} /{\rm Tr} GG^{\dagger}$.
In this way one obtains random states
generated according to the Hilbert--Schmidt measure in $\Omega_4$ \cite{ZS01,ZPNC11}.

Numerical data, obtained from $10^5$ density matrices and 
represented in Fig. \ref{two-qubits}a by closed boxes,
fit well to the analytical distribution $P_4(x)$
denoted in the plot by a solid curve.
Here  $x=N\lambda=4\lambda$ denotes the rescaled eigenvalue of $\rho$.
In general,  for any finite $N$, 
the distribution $P_N(x)$ can be represented
as a superposition of $2N$ polynomial terms,
\begin{equation}
P_N(x)=\sum_{m=2}^{2N} a_m (Nx)^{m-2} (1-Nx)^{N^2-m}  
\label{PN}
\end{equation}
with weights $a_m$ given by the Euler gamma function \cite{SZ04}.
Observe characteristic oscillations of the density
and four maxima of the distribution
related to the effect of level repulsion.

Given any bipartite quantum state $\rho$ 
it is straightforward to diagonalize the partial transpose,
 $\rho^{T_2}=({\mathbbm 1} \otimes T) \rho$,
 to check, whether this matrix is positive.
 In the case of two qubits the relative volume of the set of 
  states with positive partial transpose  (PPT)
 with respect to the HS measure is conjectured to be equal to 
 $8/33\approx 0.242$ \cite{SD15}. Thus verifying a posteriori the PPT 
 property we could divide the random states into two classes
 and obtain histograms of the level density for the sets 
of states with positive or negative partial transpose
(sets PPT and NPT, respectively) -- see. Fig. \ref{two-qubits}a.
 
 Furthermore, we studied the level density of the partially transposed 
 operators $\rho^{T_2}$ with eigenvalues $\chi$. 
For the states sampled from the entire set 
 $\Omega_4$, some partially transposed matrices are not positive, so the
 level density is supported also for negative values of the variable $y=4\chi$
 -- see. Fig. \ref{two-qubits}a. It is known that for $2 \times 2$ systems
   only a single eigenvalue $\rho^{T_2}$ can be negative 
   and it is not smaller than $-1/2$  \cite{STV98},
   so one sees a single peak of the density  $P(y)$ containing the 
   negative  eigenvalues.
 For large systems the asymptotic density for the transposed states is known
 to converge to the shifted semicircle \cite{Au12,BTL12}, 
 but no analytical expression for such a density in the case 
of $K=2$ is known so far.
 
 Figure \ref{two-qubits}b presents also histogram 
for the transposed states belonging to both classes of 
 PPT and NPT states. Note that for the set of
 PPT states the level density   $P(x)$ for the spectrum of $\rho$
 and $P(y)$ for the spectrum of $\rho^{T_2}$, are the same. 
This is a consequence of the fact that the
 partial transpose is a volume preserving involution, 
 so the level density will not be altered, if  statistics is restricted 
to PPT states only. 
Observe also that the spectrum corresponding to the PPT 
states of size $N=4$ do not display a singularity at the origin.

%

\begin{figure}
\centering
(a)\raisebox{-5cm}{\includegraphics[width=.45\textwidth]{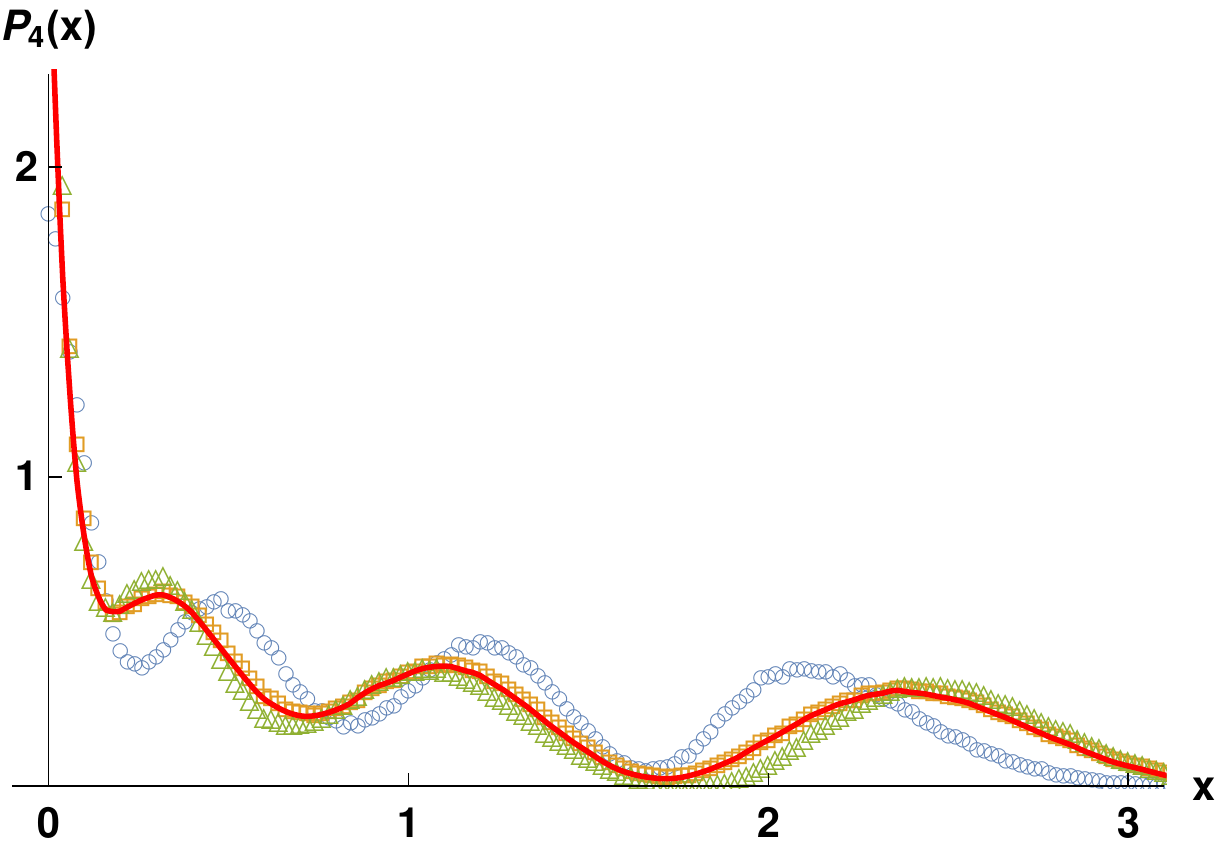}}
(b)\raisebox{-5cm}{\includegraphics[width=.45\textwidth]{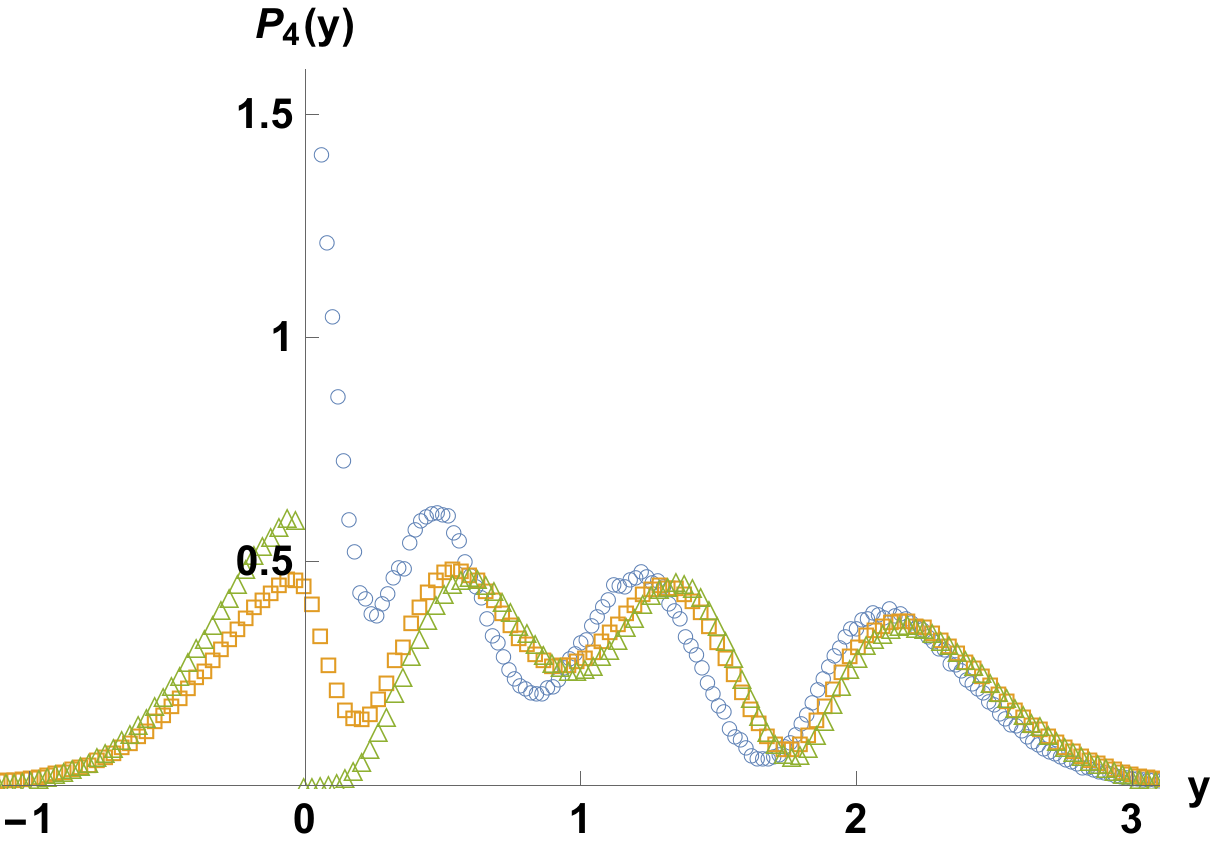}}
\caption{(a) --- Spectral density $P(x)$ for two--qubit density matrices $\rho$
where $x=4\lambda$:
 numerical data representing sampling $10^5$ times over entire set $\Omega_4$
(light/yellow $\Box$) compared with analytical distribution 
(solid, red curve).
This distribution differs from level density corresponding 
to the set of PPT states (blue $\bigcirc$), and the set of states with 
negative partial transpose states (green $\bigtriangleup$).
Panel (b) shows analogous data for partially transposed states, $\rho^{T_2}$:
(light/yellow $\Box$) denotes data for entire set $\Omega_4$,
(blue $\bigcirc$) denotes transposed PPT states, while
(green $\bigtriangleup$) denotes transposed states with 
negative partial transpose.}
\label{two-qubits}
\end{figure}

\subsection{Larger systems}

In the asymptotic case of large systems, $N >> 1$, 
the oscillations of the mean level density 
averaged over the entire set $\Omega_N$ of quantum states 
are smeared out and the distribution  (\ref{PN})
tends  to the Marchenko--Pastur distribution \cite{MP67},
\begin{equation}
\label{MPdist}
MP(x) =\frac{1}{2 \pi} \sqrt{\frac{4-x}{x}},
\end{equation} 
where the rescaled eigenvalue  $x=N \lambda$ belongs to the interval $[0,4]$. 
Although there exists pure quantum states with the operator 
norm equal to unity, $||\rho||=1$, a norm of a generic large dimension $N$ mixed state is asymptotically almost surely bounded from above by $4/N$,
since the support of the above distribution is bounded by $x_{\rm max}=4$.
Due to the effect of concentration of measure
in high dimensions a typical quantum state has spectral
distribution close to the Marchenko--Pastur law (\ref{MPdist}),
while states with different level densities become unlikely 
\cite{PPZ16}.

To demonstrate the quality of the hit and run algorithm,
we used it to generate samples of random states
for the entire convex set $\Omega_{N}$.
Diagonalising the generated states 
we obtained their eigenvalues $\lambda$,
and analyzed the level density $P(x)$ with
$x=N\lambda$. Histograms obtained in this way for $N=9$ and $N=25$
coincide with the analytical expression (\ref{PN}) --
see Fig. \ref{figTOT}.
Observe that the data for $N=25$ 
are already close to the Marchenko--Pastur distribution
valid asymptotically.
Although in this case the random points are generated
from the set $\Omega_{25}$ of dimension $d=N^2-1=624$, 
the average over a sample consisting of 
circa $10^5$ random points provides reliable data.

\begin{figure*}
    (a)\raisebox{-5cm}{ \includegraphics[width=0.45\textwidth]{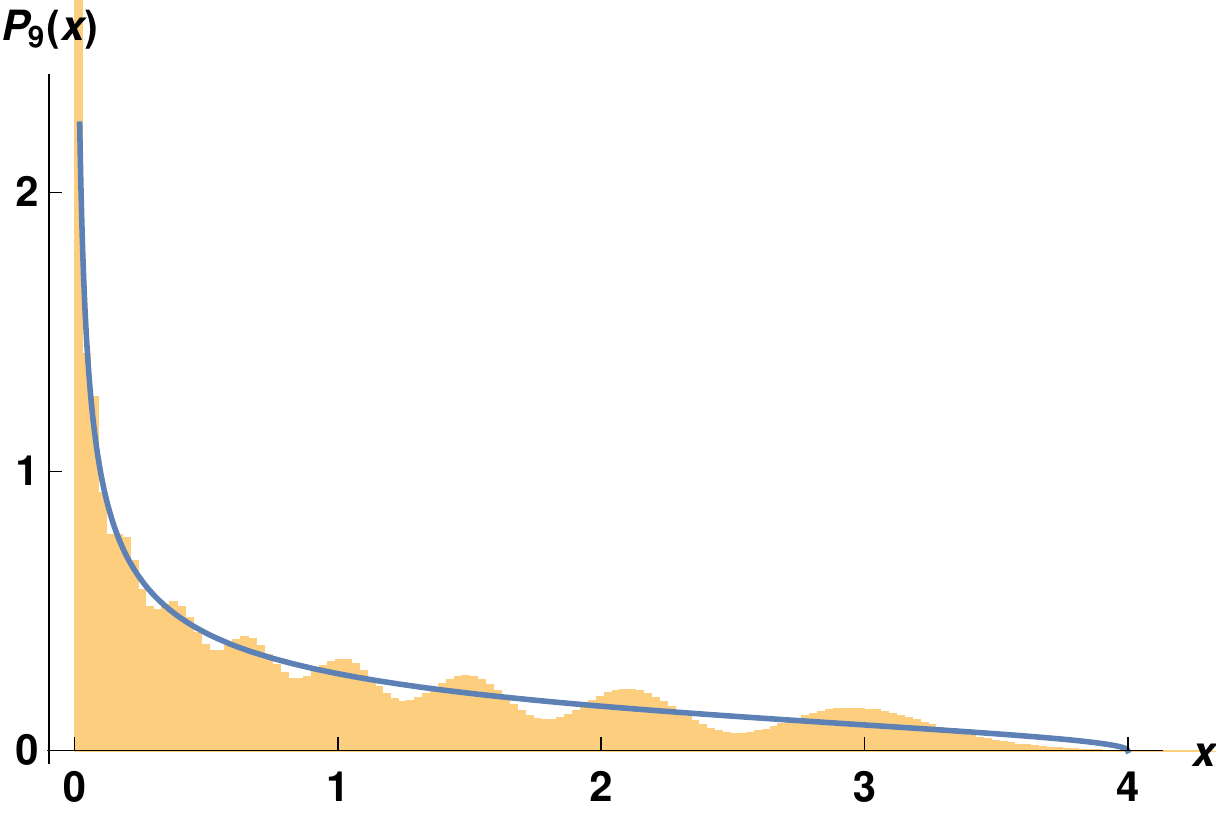}}
    (b)\raisebox{-5cm}{ \includegraphics[width=0.45\textwidth]{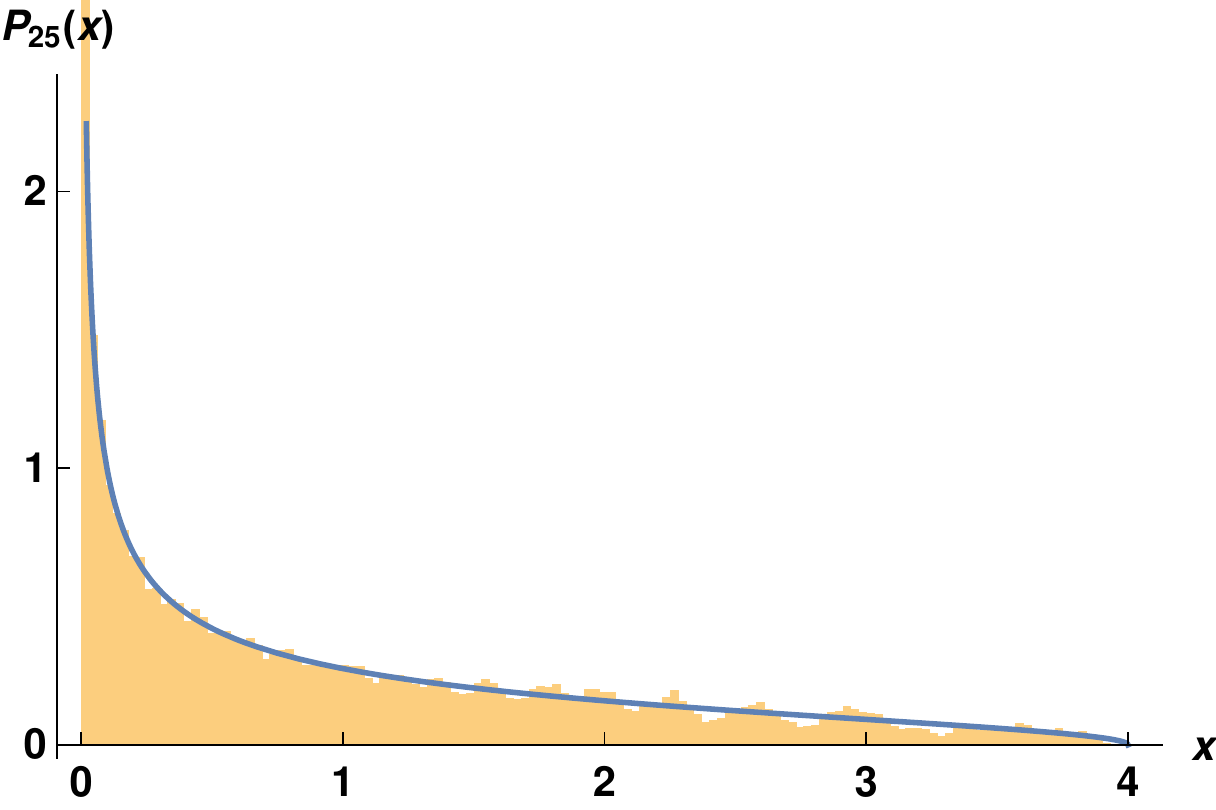}}
%
%
%
%
%
\caption{Spectral densities $P_N(x)$ averaged over entire 
set $\Omega_N$ of quantum states of size: 
a) $N=9$ and b) $N=25$.
Spectral density  (\ref{MPdist}) denoted by a solid (blue) 
line describes well the data already for $N=25$.
\label{figTOT}
}
\end{figure*}

\medskip

Consider now a bipartite case of size $N=K \times K$, for which
the notion of partial transpose is defined.
It is was first shown by Aubrun \cite{Au12}
that level density of the partially transposed bipartite random states,
written $\rho^{T_2}=({\mathbb I} \otimes T)\rho$,
asymptotically tends to the shifted semicircle law of Wigner,
\begin{equation}
P^T(x)  \sim \frac{1}{2 \pi } \sqrt{4 - (x - 1)^2} ,
\label{SSC}
\end{equation}
with support $x\in [-1,3]$.

Up to a linear shift, $x\to x-1$, 
this distribution is that of the ensemble of
random Hermitian matrices of GUE -- see e.g. \cite{Fo10}.

An extended model of random GUE matrices with spectrum of the rescaled
eigenvalue restricted to a certain interval $[0,L(z)]$ 
was analyzed by Dean and  Majumdar \cite{DM08}.
They derived an explicit family of  probability distributions,
labeled by a parameter $z$, which determines the position of the
`hard wall barrier' for the corresponding Coulomb gas model,
\begin{equation}
h_z(y)
=\frac{\sqrt{L(z)-y}}{2\pi\sqrt{y}}\left(L(z)+2y+2z\right),
\label{distDBa}
\end{equation} 
where the upper edge $L(z)$ of the support of $h_z(y)$ reads 
\begin{equation}
L(z)  =
\frac{2}{3}\left(\sqrt{z^{2}+6}-z\right)
.
\label{distL}
\end{equation}

Since partial transpose is a volume preserving involution,
the level density averaged over the set of PPT states
can be approximated by the density of states obtained 
from the shifted semicircle law of Aubrun (\ref{SSC}),
by imposing the restriction that all eigenvalues are non-negative.
The approximation consists in an assumption
that the ensemble of partially transposed density matrices
asymptotically coincides with the shifted GUE.

Taking distribution (\ref{distDBa}), setting $z=0$
and normalizing the rescaled variable $x=ay+b$  so that
the expectation value is set to unity as required,  
$\langle x \rangle = \int x g(x)=1$,
 we arrive at the normalized probability distribution
\begin{equation}
g(x)=
\frac{4}{27\pi}\sqrt{\frac{3-x}{x}}
\left(3+2x\right),
\label{eq:sp}
\end{equation}

supported in $[0,3]$.

To demonstrate that this distribution can describe
asymptotic level density of 
PPT states we need to rely on numerical computations.
In the case of a larger bi--partite system of size $N=K^2$
 the probability of finding a random PPT state decays exponentially with 
the dimension \cite{ZHSL98,GB02,AS06}.
Therefore the simplest approach to get a PPT state
by generating random states from the entire set  $\Omega_N$ 
and verifying, whether its partial transpose is positive becomes ineffective.

However, the `hit and run' algorithm, advocated in this work,
is still suitable for the case of the set of PPT states. 
For any two points inside the set it is easy to find
where the line joining them hits the boundary, provided the dimensionality $N=K^2$ is low enough. We have found the numerical procedure to be stable for $K\le5$.
Level density for the subset of PPT states for bipartite $K \times K$
system is shown in Fig. \ref{figPPT}. 
The number of random points generated, equal to 
$5\times 10^4$ for  $N=9$  and $4.4   \times 10^5$ for $N=25$, was found to give reliable results.
Characteristic finite--size oscillations visible
for $N=9$ become less pronounced for $N=25$.
As for this dimension the exact expression (\ref{PN}) 
is already close to the Marchenko--Pastur distribution, 
our data obtained for the $5 \times 5$ system
support the conjecture that the level density for the PPT states is 
asymptotically described by the distribution (\ref{eq:sp}).


\begin{figure*}
    (a) \raisebox{-5cm}{\includegraphics[width=0.45\textwidth]{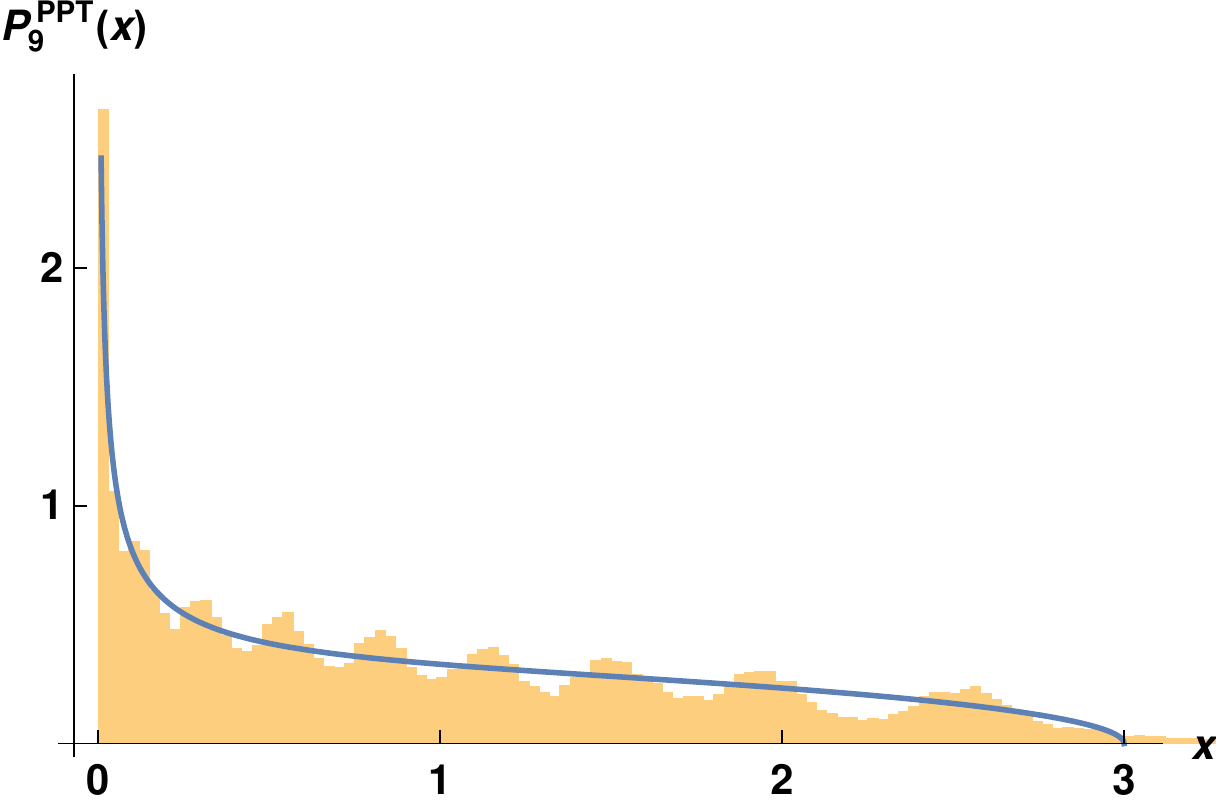}}
    (b)\raisebox{-5cm}{\includegraphics[width=0.45\textwidth]{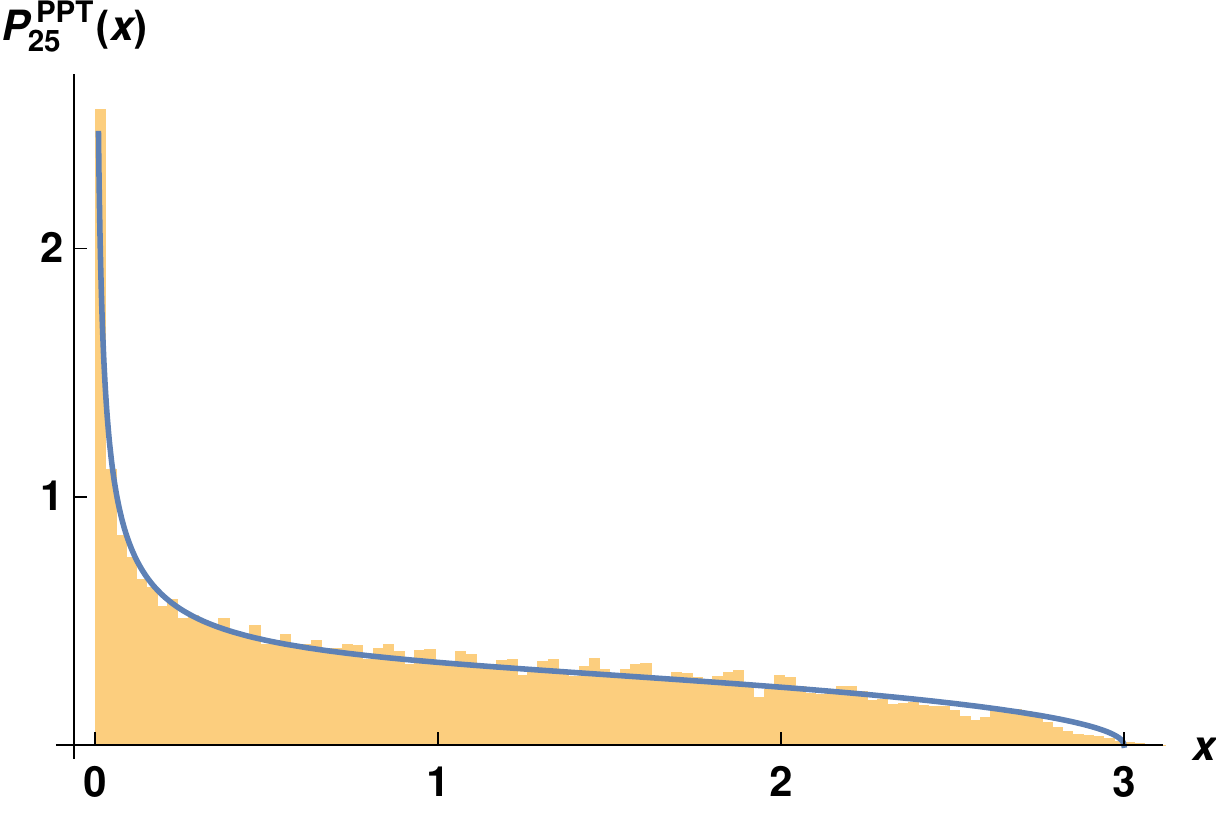}}
%
%
%
%
%
\caption{Spectral densities $P_N^{PPT}(x)$ for PPT states of bipartite systems of size: 
a) $N=3\times 3=9$ and b) $N=5 \times 5=25$. 
 Spectral density (Eq. (\ref{eq:sp})) denoted by a solid (blue) 
line describes well the data  for  $N=25$.
\label{figPPT}
}
\end{figure*}

Observe that the right edge of the support of $P_N^{\rm PPT}(x)$
is smaller than the upper bound for the MP distribution, $x_{max}=4$.
Hence our numerical results support the conjecture that the
operator norm of the PPT states is asymptotically almost surely 
smaller than the  norm of a generic state
 taken from the entire set $\Omega_{N}$,
which typically behaves as $||\rho|| \approx_{a.s.} 4/N$.

\bigskip

In view of these numerical results, it is natural to conjecture that
as $N\to\infty$, the spectral distribution of 
$\rho_{\rm PPT}$ is different from the distribution obtained
from an uniform sampling on all states. 
To prove such a statement 
one could show that there exists $\varepsilon >0$ such that
$||\rho_{\rm PPT}|| \le_{a.s.} C_{PPT}/N \leq (4-\varepsilon )/N$,
as this would result in differentiating mathematically the 
uniform distribution on all states from the uniform distribution 
on PPT states. 
However, despite our efforts, we were
 not able to prove this result and we leave it as an open question. 

It is slightly easier to deal with these sets of quantum states,
which are invariant with respect to unitary transformations.
There is a natural notion of $APPT$ (absolute positive partial transpose), which means that a state satisfies
 the $PPT$ property regardless of its (global) unitary evolution. 
In the case of $2\times 2$ systems this property is
equivalent to absolute separability \cite{KZ01},
i.e. separability with respect to any choice
of a two dimensional subspace embedded in ${\cal H}_4$,
which defines both subsystems.
Hence the property of APPT of a given state cannot 
depend on its eigenvectors but is only a function
of its eigenvalues \cite{VADM01}.
This feature holds also for higher dimensions 
and in some cases the boundary of this set are known \cite{Hi07,Jo13}.
Interestingly, although for higher dimensions  the set of PPT
states has much larger volume then the set of separable states
\cite{SWZ08}, it is conjectured that the set of absolutely
separable states and APPT states do coincide \cite{AJR15}.
Furthermore, Proposition 8.2 from \cite{JLNR15} concerning the 
largest eigenvalue of a quantum state belonging to the set of APPT states implies
that the support of the level density 
for the states of this set is asymptotically bounded by $x=3$.
Observe that this fact is consistent with our observations
concerning the generic behavior of the norm of a PPT state.

Note that the above questions of separating distributions through their typical largest eigenvalues 
can be partly addressed in the more general context of sampling from a uniform purification. 
It is  known \cite{ZS01,ZPNC11}
that if the ancilla space has the same dimension as the state space, 
this sampling method gives the uniform distribution. 
It is also interesting to study other regimes, where the ancilla space is bigger than the state space.  Aubrun 
proved in \cite{Au12} that the PPT property holds with probability $1$ as 
$N\to\infty$ if, for $\varepsilon >0$ the ancilla space
is of dimension at least $(4+\varepsilon) N$. Intuitively, a bigger ancilla space means that the sampling is more concentrated
around the maximally mixed state. 
For even larger ancillas 
the generated states belong to the set of APPT states,  
as it was showed in \cite{CNY12} that the threshold is $(4+\varepsilon) N^2$ 
and that this is essentially optimal.
Unsurprisingly, this shows that the set of APPT states is much smaller than the set of PPT states. 
But the problem of comparing these sampling probabilities with the uniform PPT distribution
or the uniform APPT distribution remains difficult to achieve formally.

\section{Concluding Remarks}\label{section:remarks}

In this paper we analysed a universal algorithm to generate random points
inside an arbitrary compact set $X$ in ${\mathbb R}^d$
according to the uniform measure.
Any initial probability measure $\mu$ transformed
by the corresponding Markov operator converges
exponentially to the invariant measure $\mu_*$,
uniform in $X$. Explicit estimations for the convergence rate are derived
in terms of the ratio $\kappa=r/R$ between the radii of the sphere
inscribed inside $X$ and the sphere outscribed on it.

Thus the algorithm presented here can be used in practice
to generate, for instance, a sample
of random quantum states. In the case of
states of a composed quantum system, one can also
generate a sequence of random states with positive partial
transpose. Sampling random states satisfying a given condition
and analyzing their statistical properties is relevant
in the research on quantum entanglement and
correlations in multi-partite quantum systems.
A standard approach of generating random points
from the entire set of quantum states with respect to the flat measure \cite{ZS01}
and checking a posteriori, whether the partial transpose of the state constructed
is positive, becomes inefficient for large dimensions,
as the relative volume of the set of PPT states becomes
exponentially small \cite{ZHSL98}.

Obtained numerical results show that the
level density for random states covering uniformly
the subset of the PPT states
differs considerably from the Hilbert-Schmidt level density 
corresponding to the entire set  $\Omega_{K^2}$
of mixed states of a bipartite system.
Making use of the shifted semicircle law of Aubrun (\ref{SSC}),
which describes the asymptotic density of partially transposed states,
and imposing the restriction that all the eigenvalues 
are non-negative we arrived at
the probability distribution (\ref{eq:sp}).
This distribution was compared  with the level density 
obtained numerically with the `hit and run' algorithm 
applied for the set of PPT states for the $K \times K$ system
for $K=3,4,5$. The larger dimension
the better agreement of the numerical data
with the distribution (\ref{eq:sp}),
so is tempting to conjecture that it 
describes the level density of the PPT states in the asymptotic limit.
%

\medskip

\section*{Acknowledgements}

We are obliged to Uday Bhosale and
Arul Lakshminarayan for their interest
in the project and several fruitful discussions.
It is a pleasure to thank Nathaniel Johnston
and Ion Nechita for constructive remarks and for bringing
to our attention Ref. \cite{DM08}.
The research of TS was supported by 
the National Science Center of Poland, grant number DEC- 2012/07/B/ST1/03320 and EU grant RAQUEL,
while K{\.Z} acknowledges a support by the NCN grant DEC-2011/02/A/ST1/00119.
BC was supported by JSPS Kakenhi grants number 26800048 and 15KK0162,  and  the grant number ANR-14-CE25-0003.


\begin{thebibliography}{99}
\bibitem{AS03} E. M. Alfsen and F. W. Shultz,
 {\sl Geometry of State Spaces of Operator Algebras},
 Birkh{\"a}user, Boston  (2003).

\bibitem{ACH93} M. Adelman, J. V. Corbett and C.~A.~Hurst,
 The geometry of state space,
 {\sl Found. Phys.} {\bf 23}, 211 (1993).

\bibitem{KZ01} M. Ku{\'s} and K. {\.Z}yczkowski,
Geometry of entangled states.
 {\sl Phys. Rev.} {\bf A~63}, 032307 (2001).

 \bibitem{GMK05}   J. Grabowski, G. Marmo and M. Ku{\'s},
 Geometry of quantum systems: density states and entanglement.
  {\sl J. Phys.} {\bf  A 38}, 10217-10244 (2005).


\bibitem{BZ06} I.~Bengtsson and K.~\.{Z}yczkowski, 
{\sl Geometry of Quantum States: An
Introduction to Quantum Entanglement}.
 Cambridge University Press, Cambridge (2006).


\bibitem{H^4} R. Horodecki, P. Horodecki, M. Horodecki and  K. Horodecki,
Quantum entanglement,
{\sl Rev. Mod. Phys.} {\bf 81}, 865 (2009).

\bibitem{HHH96a} M. Horodecki, P. Horodecki and R. Horodecki,
 Separability of mixed states: necessary and sufficient conditions.
  {\sl Phys. Lett.} {\bf A~223}, 1 (1996).

\bibitem{ZS01} K. {\.Z}yczkowski and H.-J. Sommers,
   Induced measures in the space of mixed quantum states,
{\sl J. Phys.} {\bf A 34}, 7111-7125 (2001).

\bibitem{ZPNC11} K. {\.Z}yczkowski, K. A. Penson, I. Nechita, B. Collins,
Generating random density matrices,
{\sl J. Math. Phys.} \textbf{52}, 062201 (2011).

\bibitem{ZHSL98} K. {\.Z}yczkowski, P. Horodecki, A. Sanpera and
M.~Lewenstein,  Volume of the set of separable  states,
  {\it Phys. Rev.} {\bf A58}, 883-892 (1998).

\bibitem{GB02}  L. Gurvits and H. Barnum, 
 Largest separable balls around the maximally mixed bipartite 
quantum  state.
{\sl Phys. Rev.} {\bf A 66}, 062311 (2002).


\bibitem{AS06} G. Aubrun and S. J. Szarek,  Tensor products of
convex sets  and the volume of separable states on $N$ qudits.
{\sl Phys. Rev.} {\bf A 73}, 022109 (2006).

\bibitem{Au12} G.~Aubrun,
 Partial transposition of random states
   and non-centered semicircular   distributions,
  {\sl Random Matrices Theor. Appl.} {\bf 1}, 1250001 (2012).
 
 \bibitem{BTL12} U.T. Bhosale, S. Tomsovic and A. Lakshminarayan,
 Entanglement between two subsystems, the Wigner semicircle and extreme value 
 statistics,
 {\sl Phys. Rev.} {\bf  A~85}, 062331 (2012).
 
\bibitem{DM08} D.S. Dean and S. N. Majumdar, 
Extreme value statistics of eigenvalues of Gaussian random matrices,
{\sl Phys. Rev.} {\bf E~77}, 041108 (2008).


\bibitem{Lo93} L. Lov\' asz, Random walks on graphs: a survey,
 in Combinatorics:
Paul Erd\" os is eighty (Keszthely, Hungary, 1993), vol. 2,
pp. 353-397, edited by 
D.~Mikl\' os et al., 
Bolyai Soc. Math. Stud. 2, 
Budapest, 1996.

\bibitem{Ve05} S. Vempala, Geometric Random Walks: A Survey, 
{\sl Combinatorial and Computational Geometry, MSRI Publications}
{\bf 52}, 573-612 (2005)


\bibitem{LV1}  L. Lov\'asz and S. Vempala,  
Hit-and-run from a corner,
{\sl  SIAM J. Comput.} {\bf 35}, 985--1005. (2006).

\bibitem{LV2}  L. Lov\'asz, and S. Vempala,
          The geometry of logconcave functions and sampling algorithms,
         {\sl  Random Structures Algorithms} {\bf 30}, 307--358 (2007).


\bibitem{CKKSZ13}  B. Collins, T. Kousha, R. Kulik, T. Szarek,  and K. {\.Z}yczkowski, 
 Exponentially convergent algorithm to generate random 
points in a $d$--dimensional body, 
 preprint arXiv:1312.7061 and J. Convex Analysis, {\sl 2016, in press}

\bibitem{Gu04} L. Gurvits, Classical complexity and quantum entanglement, 
{\sl Journal of Computer and System Sciences} {\bf 69}, 448-484  (2004).
\bibitem{MP67} V. A. Marchenko and L. A. Pastur, 
 The distribution of eigenvalues in certain sets of random matrices,
  {\sl Math. Sb.} {\bf 72}, 507 (1967).


\bibitem{SBZ06} S. Szarek,  I. Bengtsson and K. {\.Z}yczkowski,
On the structure of the body of states with positive partial transpose,
{\sl J. Phys.} {\bf  A 39}  L119-L126 (2006).

\bibitem{Has} W. K. Hastings, 
Monte Carlo sampling methods using Markov chains and their applications,
{\sl Biometrika}, {\bf 57},  97-109 (1970).

\bibitem{Doeblin} W. Doeblin,  Sur les propri\'et\'es asymptotiques de
mouvement r\'egis par certains types de cha\^ines simples, 
{\sl Bull. Math. Soc. Roum. Sci.} {\bf 39}, 57--115 (1937).


\bibitem{SZ04} H.--J. Sommers,  K. {\.Z}yczkowski, 
Statistical properties of random density matrices,
{\sl J. Phys.} {\bf A 37}, 8457 (2004).


\bibitem{SD15} P. B. Slater and C. F. Dunkl,
 Formulas for Rational-Valued Separability Probabilities of
 Random Induced Generalized Two-Qubit States,
{\sl Advances Math. Phys.} {\bf 2015}, 621353.


 \bibitem{STV98} A. Sanpera, R. Tarrach and G. Vidal,
   Local description of quantum inseparability,  
  {\sl Phys. Rev.} {\bf  A 58}, 826 (1998).


\bibitem{PPZ16} Z. Pucha{\l}a, {\L}. Pawela, K. {\.Z}yczkowski,
Distinguishability of generic quantum states,
{\sl Phys. Rev.} {\bf A 93}, 061221 (2016).

 
 \bibitem{Fo10}  P. J. Forrester, 
 {\sl Log-gases and Random matrices},
 (Princeton University Press, Princeton, 2010).


\bibitem{VADM01} F. Verstraete, K. Audenaert, and B. DeMoor,
   Maximally entangled mixed states of two qubits,
{\sl Phys. Rev.} {\bf A 64}, 012316 (2001).

\bibitem{Hi07} R. Hildebrand,
Positive partial transpose from spectra,
{\sl  Phys. Rev. A} 76, 052325 (2007).

\bibitem{Jo13} N. Johnston,
Separability from spectrum for qubit--qudit states,
{\sl Phys. Rev.} {\bf  A 88}, 062330 (2013).

\bibitem{SWZ08}   S. J. Szarek, E. Werner, and  K. {\.Z}yczkowski, 
Geometry of sets of quantum maps: a generic positive map acting on a
  high-dimensional system is not completely positive,
{\sl J. Math. Phys.} {\bf 49}, 032113-21  (2008).

\bibitem{AJR15} S. Arunachalam, N. Johnston,  and V. Russo,
 Is separability from spectrum determined by the partial transpose? 
{\sl Quant. Inf. Comput. 15} 0694-0720, (2015).

\bibitem{JLNR15}   M. A. Jivulescu, N. Lupa, I. Nechita and D. Reeb,
Positive reduction from spectra
Lin. Alg. Appl. 469, 276-304 (2015).

\bibitem{CNY12}  B. Collins, I. Nechita and D. Ye,
 The absolute positive partial transpose property for random induced states,
 {\sl Random Matrices Theor. Appl.} {\bf 1}, 1250002 (2012).




\end{thebibliography}
\end{document}